\documentclass[11pt,twoside, final]{elsarticle}
\usepackage{etoolbox,lastpage}
\date{\scriptsize   Received: , Accepted: .}
\usepackage{amsmath,amsthm,amscd,amsfonts,amssymb,enumerate}
\usepackage{graphicx}		
\usepackage{color}
\usepackage[colorlinks]{hyperref}
\usepackage{verbatim}
\usepackage{caption}
\usepackage{stmaryrd}
\usepackage[english]{babel}
\newtheorem{thm}{{\bf Theorem}}[section]
\newtheorem{cor}{{\bf Corollary }}[section]
\newtheorem{lem}{{\bf Lemma}}[section]
\newtheorem{prop}{{\bf Proposition}}[section]

\newtheorem{defn}{Definition}[section]
\newtheorem{rem}{{\bf Remark}}[section]
\numberwithin{equation}{section}
\newcommand{\R}{\mathbb{R}}
\newcommand{\N}{\mathbb{N}}

\newcommand{\C}{\mathbb{C}}

\newcommand{\ds}{\displaystyle}

\everymath{\displaystyle}
\newcommand{\ef}{\quad \hfill$\square$\vspace{0.15cm}\\}

\begin{document}
\title{Quantum algebra from generalized $q$-Hermite polynomials }

\author{Kamel Mezlini}
\address{Universit\'e de Tunis El-Manar, Facult\'e des
Sciences de Tunis,\\ D\'epartement de  Math\'ematiques, 2092 Tunis El-Manar, Tunisie.
\\ email: kamel.mezlini@lamsin.rnu.tn}

\author{Najib Ouled Azaiez}
\address{
Current address: Department of Mathematics, College of Science,\\
King Faisal University, P. O. Box 400\\
Al-Ahsa 31982, Saudi Arabia. \\
email:  nazaiez@kfu.edu.sa \\
Address: Universit\'e de Sfax, Facult\'e  des Sciences de Sfax,\\
D\'epartement de  Math\'ematiques,  B.P 1171, Sfax 3000, Tunisie.}

\date{\today}

\begin{abstract}
In this paper, we  discuss new results related to  the generalized discrete $q$-Hermite
II polynomials $ \tilde h_{n,\alpha}(x;q)$, introduced by
Mezlini et al. in 2014. Our aim is to give a continuous orthogonality relation,
a $q$-integral representation and an evaluation at unity of the Poisson kernel,
for these polynomials.
Furthermore, we introduce $q$-Schr\"{o}dinger operators and
give  the  raising and lowering operator algebra
corresponding to  these polynomials.
 Our results generate a new explicit realization  of the quantum algebra
$\mathsf{su}_{q}(1, 1)$, using the generators  associated with a $q$-deformed  generalized
para-Bose oscillator. \\
\textbf{keywords:} $q$-orthogonal polynomials, $q$-deformed algebras,
harmonic oscillators. 
\end{abstract}
\maketitle

\section{Introduction}

The $q$-deformed harmonic oscillator  algebras  have been intensively studied
in recent years due to their crucial role in diverse areas of mathematics and physics
(see \cite{ Drinfel,Floreanini,Kulish, Macfarlane1}).
One of the most important applications of   $q$-deformed algebras based theory  arises
from a generalization of the fundamental symmetry concept of the  classical  Lie algebras.

Many  algebraic constructions were proposed in the literature to describe assorted
generalizations of the quantum harmonic oscillator. However, a common difficulty
for most of them is to derive an explicit form of  associated Hamiltonian eigenfunctions.
It is well known that   Hermite  polynomials are connected to the realization of
classical-harmonic-oscillator algebras. It is worth to mention that  generalizations
of quantum harmonic oscillators  lead to generalizations of $q$-Hermite polynomials.
An explicit realization of  $q$-harmonic oscillator has been explored by many  authors
(see for instance  \cite{Atakishiyev, Atakishiyev2, Borzov, Kulish}),
where the eigenfunctions of the corresponding  Hamiltonian are given explicitly in terms
of $q$-deformed Hermite polynomials.
Generators of the corresponding oscillator algebra are realized in terms of
first-order difference operators.\\
This paper investigates the generalized discrete $q$-Hermite II polynomials
$ \tilde h_{n,\alpha}(x;q)$  to construct a new realization  of  quantum algebra.
From this generalization, we obtain an  explicit
form of the generators for  quantum  algebra, in terms of $q$-difference operators.\\

The structure of this paper is as follows:
Sec. 2 describes briefly  the main definitions and  properties of some $q$-basic special
functions and operators \cite{andrews,G}; Sec. 3 recalls some notations and useful
results about the generalized discrete $q$-Hermite II polynomials \cite{ghermite}.
Therefore, we obtain   continuous orthogonality relations.
Moreover, an integral representation of generalized discrete
$q$-Hermite II-polynomials is proposed, and an evaluation at unity of the  Poisson kernel
for a
family of polynomials $ \tilde h_{n,\alpha}(x;q)$, is also studied. In addition to this, for
$\alpha = 1/2$,  a formula using $q$-trigonometric functions
$Cos_q(x)$, $ Sin_q(x)$, and  an expression for the second $q$-Bessel functions in terms of
the generalized discrete $q$-Hermite II-polynomials, is given. Among other things,
%%%%%%%%%%%%%%%%%%%%%%%%%%%%%%%%%%   15/04/19   %%%%%%%%%%%%%%%%%%%%%%%%%%%%%%5
by the specialization $x = i q^{\alpha + 1/2}$ in  the generating function of the even  generalized $q$-Hermite polynomials,
a special case of the  Heine  transformation of $\ _{2}\phi_1$ series formula, is recovered. 
Sec. 4 provides an explicit new realization  of  quantum algebra, in which the generators are
associated with $q$-deformed  generalized para-Bose oscillator.

\section{Notations and Preliminaries}
This section is systematically organized  in the following  order; Sec. 2.1 introduces
some basic notations; Sec. 2.2 recalls  the definitions
of $q$-derivatives and $q$-integrals; Sec. 2.3 recalls
 the definition of some $q$-special functions that are important in our paper.
\subsection{Basic symbols}
For the convenience of the reader, we provide in this section a summary of the
mathematical notations and definitions, see \cite{G, ghermite}.
Throughout this paper, we assume that $0<q<1$ and $\alpha > -1.$
For each complex number $a,$ we  define $q$-shifted factorials, being
$$(a;q)_{0}=1;\ \ \ (a;q)_{n}=\prod_{k=0}^{n-1}(1-aq^{k}),
n=1, 2,...;\ \ \ (a;q)_{\infty}=\prod_{k=0}^{\infty}(1-aq^{k}).$$
The $q$-number and the  $q$-factorial are defined as follows:
\begin{equation*}
\llbracket  x\rrbracket_{q}={{1-q^x}\over{1-q}}, ~~~~~ x\in \C\quad {\rm and}\quad n!_q =\llbracket
1\rrbracket_{q}\llbracket  2\rrbracket_{q}...\llbracket  n\rrbracket_{q},   ~~~~~~  0!_{q}=1,  ~~~~~~n\in \N.
\end{equation*}
For each  real $\alpha>-1,$ the  generalized $q$-integers and the generalized $q$-factorials
are defined as:
\begin{equation}\label{nfactalpha}
 \begin{array}{lll}
\left\llbracket2n \right\rrbracket_{q,\alpha}   &=& \left\llbracket 2n\right\rrbracket_{q}, \\
\left\llbracket 2n+1\right\rrbracket_{q,\alpha} &=& \left \llbracket2n+2\alpha+2\right\rrbracket_{q},\\
n!_{q,\alpha}&=&\left\llbracket1 \right\rrbracket_{q,\alpha}\left\llbracket2 \right\rrbracket_{q,\alpha}...\left\llbracket
n \right\rrbracket_{q,\alpha},
\;\;0!_{q,\alpha}=1,
\end{array}
\end{equation}
and  the generalized $q$-shifted factorials are defined as:
\begin{equation}\label{qqnalpha}
(q;q)_{n,\alpha}:=(1-q)^nn!_{q,\alpha}.
\end{equation}
We can  write (\ref{qqnalpha})  explicitly as:
\begin{equation}\label{ecr-fac}
\begin{array}{lll}
(q;q)_{2n,\alpha}&=&(q^2;q^2)_{n}(q^{2\alpha+2};q^2)_{n},\\
(q;q)_{2n+1,\alpha}&=&(q^2;q^2)_{n}(q^{2\alpha+2};q^2)_{n+1}.
\end{array}
\end{equation}
\begin{rem}\label{relaqqalpha}
The specific value $\alpha=-\frac{1}{2}$ leads to
$(q;q)_{n,\alpha}= (q;q)_{n}\;\;\mbox{and}\;\;n!_{q,\alpha }=n!_{q}$.
\end{rem}
\subsection{$q$-derivatives and  $q$-integral }

\noindent  Jackson's $q$-derivative  $D_{q}$ (see \cite{G,KC})
is defined by :
\begin{equation}\label{jacksonD}
 D_{q}f(z)=   \frac{f(z)-f(qz)}{(1-q)z}.
\end{equation}
\noindent The variant $D_q^+$, called  forward
$q$-derivative of the (backward) $q$-derivative $D_q^- = D_q$
(\ref{jacksonD}), is defined as:
\begin{equation}\label{jacksonD+}
  D_{q}^+f(z)=   \frac{f(q^{-1}z)-f(z)}{(1-q)z}.
\end{equation}
\noindent Note that $  \lim _{q\rightarrow 1^- }D_qf(z)=\lim
_{q\rightarrow 1^- }D_q^+f(z) =f^{\prime}(z)$
whenever $f$ is differentiable at $z$.\\
Generalized backward and forward $q$-derivative operators $D_{q,\alpha}$ and
$D_{q,\alpha}^+ $ are defined as (see \cite{ghermite})
\begin{equation}\label{D_qalpha}
 D_{q,\alpha}f(z)=   \frac{f(z)-q^{2\alpha+1}f(qz)}{(1-q)z},
\end{equation}
\begin{equation}\label{D_qalphap}
 D_{q,\alpha}^+f(z)=  \frac{f(q^{-1}z)-q^{2\alpha+1}f(z)}{(1-q)z}.
\end{equation}
Generalized $q$-derivatives operators are given by
\begin{equation}\label{defdeltaalpha}
\Delta _{q,\alpha}f =D_{q}f_e+D_{q,\alpha}f_o,
\end{equation}
\begin{equation}\label{defdeltaalphaplus}
\Delta _{q,\alpha}^+f =D_{q}^+f_e+D_{q,\alpha}^+f_o,
\end{equation}
where $f_e$ and $f_o$ are respectively the even and the odd parts of $f$.\\ For
$\alpha=-\frac{1}{2}$, we have  $\;D_{q,\alpha}= D_{q}$, $\;\;D_{q,\alpha}^+= D_{q}^+$,
$\;\;\Delta_{q,\alpha}= D_{q}\;$ and $\;\Delta_{q,\alpha}^+= D_{q}^+$.\\
We shall  need the  Jackson $q$-integral  defined by (see
\cite {G,KC}).
\begin{equation*}\label{eq:integ}
 \int_0^{\infty}{f(x)d_qx}
=(1-q)\sum_{n=-\infty}^{\infty}{f(q^n)q^n},
\end{equation*}
\begin{equation*}\label{eq:integral}
\int_{-\infty}^{\infty}{f(x)d_qx}
=(1-q)\sum_{n=-\infty}^{\infty}q^n f(q^n)+
(1-q)\sum_{n=-\infty}^{\infty}q^nf(-q^n).
\end{equation*}
One can easily show   that  these integrals converge for a bounded function $f,$
since the geometric series converges for $ 0 < q < 1$.

\subsection{Some $q$-analogues of special functions }
\noindent
The two Euler's   $q$-analogues of the  exponential
function are given by (see \cite{G})
\begin{equation}\label{q-Expo}
E_q(z)
=\displaystyle\sum_{k=0}^{\infty}\displaystyle\frac{q^{\frac{k(k-1)}{2}}z^{k}}{(q;q)_{k}}=(-z;q)_{\infty},
\end{equation}
\begin{equation}\label{q-expo}
e_q(z)
=\displaystyle\sum_{k=0}^{\infty}\displaystyle\frac{z^{k}}{(q;q)_{k}}=
\displaystyle\frac{1}{(z;q)_{\infty}},\;\;\;\;|z|<1.
\end{equation}
Then  $q$-analogues of the trigonometric functions are defined as
\begin{equation}\label{q-trigo}
Cos_q(z)=\frac{E_q(iz)+E_q(-iz)}{2},\;\;\;Sin_q(z)=\frac{E_q(iz)-E_q(-iz)}{2i}.
\end{equation}
 The generalized $q$-exponential function is defined as (see \cite{ghermite})
\begin{equation}\label{q-Expo-alpha}
E_{q,\alpha}(z):=\displaystyle\sum_{k=0}^{\infty}\displaystyle\frac{q^{\frac{k(k-1)}{2}}z^{k}}{(q;q)_{k,\alpha}}.
\end{equation}
Using Remark \ref{relaqqalpha}, the specific value $\alpha=-\frac{1}{2}$ leads to
$E_{q,\alpha}(z) = E_{q}(z)$.\\
The following  $q$-Bessel functions  can be expressed using  generalized
$q$-shifted factorials. The Jackson second $q$-Bessel function is defined as
(see \cite{G,Koekoek})
\begin{equation}\label{q-bessel2}
J_\alpha^{(2)}(x;q^2)=\frac{(q^{2\alpha+2};q^2)_\infty}{(q^{2};q^2)_\infty}\left(\frac{x}{2}\right)^\alpha
\ds\sum_{n=0}^{\infty}\frac{(-1)^nq^{2n(n+\alpha)}}{(q;q)_{2n,\alpha}}\left(\frac{x}{2}\right)^{2n}.
\end{equation}
The Hahn-Exton $q$-Bessel function is defined as (see \cite[page 20, Formula (0.7.15)]{Koekoek})
\begin{equation}\label{q-bessel3}
J_\alpha^{(3)}(x;q^2)=\frac{(q^{2\alpha+2};q^2)_\infty}{(q^{2};q^2)_\infty}
x^\alpha\ds\sum_{n=0}^{\infty}\frac{(-1)^nq^{n(n+1)}}{(q;q)_{2n,\alpha}}x^{2n}.
\end{equation}
The modified  $q$-Bessel function is defined as
\begin{equation}\label{jalpha}
j_\alpha(x;q^2)=\frac{(q^{2};q^2)_\infty}{(q^{2\alpha+2};q^2)_\infty} x^{-\alpha}J_\alpha^{(3)}(x;q^2)=
\ds\sum_{n=0}^{\infty}\frac{(-1)^nq^{n(n+1)}}{(q;q)_{2n,\alpha}}x^{2n}.
\end{equation}
\section{The generalized discrete $q$-Hermite II polynomials }
The generalized discrete $q$-Hermite II polynomials $\{\tilde h_{n,\alpha}(x;q)\}_{n=0}^\infty$ are
introduced by the first author et al. \cite{ghermite}.
We recall their definition and some of their main properties.
They are defined as
\begin{equation}\label{hnalphatild}
\tilde h_{n,\alpha}(x;q):=(q;q)_{n}\ds\sum_{k=0}^{[\frac{n}{2}]}\ds\frac{(-1)^kq^{-2nk}q^{k(2k+1)}x^{n-2k}}{(q^2;q^2)_{k}(q;q)_{n-2k,\alpha}},
\end{equation}
where $[x]$ denotes the integral part of $x\in\mathbb{ R}$.\\
\noindent For $\alpha=-\frac{1}{2}$,  $\tilde  h_{n,\alpha}(x;q)$ reduces to the discrete $q$-Hermite II
polynomial $\tilde h_{n}(x;q)$ (see \cite{Koekoek}).\\
 They  have the following  properties (see \cite{ghermite}): \\
{\bf The generating function:}
\begin{equation}\label{gghntild}
e_{q^2}(-z^2)E_{q,\alpha}( xz)=\ds\sum_{n=0}^\infty\frac{q^{\frac{n(n-1)}{2}}}{(q;q)_n}\tilde  h_{n,\alpha}(x;q)z^n.
\end{equation}
{\bf The inversion formula:}
\begin{equation}\label{monometild}
x^n=(q;q)_{n,\alpha}\ds\sum_{k=0}^{[\frac{n}{2}]}\frac{q^{-2nk+3k^2}\tilde  h_{n-2k,\alpha}(x;q)}{(q^2;q^2)_{k}(q;q)_{n-2k}  }.
\end{equation}
{\bf Forward shift operator:}
\begin{equation}\label{forwardshifttild}
\tilde h_{n,\alpha}(q^{-1}x;q)-q^{(2\alpha+1)\theta_{n+1}}\tilde h_{n,\alpha}(x;q)=q^{-n}(1-q^n)x\tilde h_{n-1,\alpha}(x;q),
\end{equation}
 where $\theta_{n}$ is defined to be $0$ if $n$ is odd and $1$ if $n$ is even.\\
{\bf Backward shift operator:}
\begin{equation}\label{backward}
\tilde h_{n,\alpha}(x;q) -q^{(2\alpha+1)\theta_{n+1}}(1+q^{-2\alpha-1}x^2)\tilde h_{n,\alpha}(qx;q)=
-q^{n}\frac{ 1-q^{-n-1-(2\alpha+1)\theta_{n}}}{1-q^{-n-1}} x\tilde h_{n+1,\alpha}(x;q).
\end{equation}
%%%%%%%%%%%%%%%%%%%%%%%%%%%%%%%%%%%%%%%%%%27/5/16%%%%%%%%%%%%%%%%%%%%%%%%%%%%%%%%%%%%%%%%%%%%%%%%%%%%%%%%5
{\bf Three terms recursion formula:}
\begin{equation}\label{recursiontild}
x\tilde  h_{n,\alpha}(x;q)-q^{1-2n}(1- q^{n})\tilde  h_{n-1,\alpha}(x;q)=
\frac{1-q^{n+1+(2\alpha+1)\theta_{n}}}{1-q^{n+1} }\tilde  h_{n+1,\alpha}(x;q).
\end{equation}
{\bf $q$-Difference equations:}
\begin{equation}\label{h2n-q-diff}
(1+q^{-2\alpha-1}x^2)\tilde h_{2n,\alpha}(qx;q)  -(1+q^{-2\alpha}+q^{2n-2\alpha-1}x^2)\tilde h_{2n,\alpha}(x;q)
+ q^{-2\alpha}\tilde h_{2n,\alpha}(q^{-1}x;q)=0
\end{equation}
and
\begin{equation}\label{h2n1-q-diff}
\hspace{-0.5cm}(1+q^{-2\alpha-1}x^2)\tilde h_{2n+1,\alpha}(qx;q) -(q+q^{-2\alpha-1}+q^{2n-2\alpha}x^2)\tilde
h_{2n+1,\alpha}(x;q)+q^{-2\alpha}\tilde h_{2n+1,\alpha}(q^{-1}x;q)=0.
\end{equation}

%%%%%%%%%%%%%%%%%%%%%%%%%%%%%%%%%%%%%%%%%%%%%%%%%%%%%%%%%%%%%%%%%%%%%%%%%%%%%%%%%%%%%%%%%%%%%%%%%%%%%%%%%%

The  family of generalized discrete $q$-Hermite II polynomials  satisfy two kind of orthogonality
relations, a discrete one and a continuous one.
As was shown in  \cite{ghermite}, we have:\\
{\bf{ A discontinuous orthogonality relation:}}\\

\begin{equation}\label{orthohnalpht}
\hspace{-0.5cm}\ds\int_{-\infty}^{\infty}\tilde h_{n,\alpha}(x;q) \tilde h_{m,\alpha}(x;q)
\omega_{\alpha}(x;q)|x|^{2\alpha+1}d_qx=
\frac{2(1-q)(-q, -q,q^2;q^2)_\infty q^{-n^2}\left(q;q\right)_{n}^2}
{(-q^{-2\alpha-1}, -q^{2\alpha+3},q^{2\alpha+2};q^2)_\infty\left(q;q\right)_{n,\alpha}}\delta_{n,m},
\end{equation}
where
\begin{equation}\label{omega}
\omega_\alpha(x;q)=e_{q^2}(-q^{-2\alpha-1}x^2)
\end{equation}
and $\delta_{n,m}$ is the Kronecker symbol.
\subsection{A continuous orthogonality relation}
Our primary interest in this paper is to prove a  continuous orthogonality relation for the
family of generalized discrete $q$-Hermite II polynomials.
First, we  rewrite the  $q$-Laguerre polynomials  (see \cite{Koekoek}) by means of the
generalized $q$-shifted factorials as follows:
\begin{equation*}
L_{n}^{(\alpha)}(x;q^2)=(q^{2\alpha+2};q^2)_n\ds\sum_{k=0}^{n}\ds\frac{(-1)^kq^{2k(k+\alpha)}x^{k}}{(q;q)_{2k,\alpha}(q^2;q^2)_{n-k}}.
\end{equation*}
The discrete $q$-Hermite II polynomials $\tilde h_{n,\alpha}(x;q)$ can also be expressed in terms
of $q$-Laguerre polynomials
$ L_{n}^{(\alpha)}(x;q)$ as follows (see \cite{ghermite}):
\hspace{-1cm}\begin{equation}\label{exp-hn-ln}
\left \{\begin{array}{lll}
\tilde h_{2n,\alpha}(x;q)  &=&   (-1)^nq^{-n(2n-1)}\frac{(q;q)_{2n}}{(q^{2\alpha+2};q^2)_{n}}L_{n}^{(\alpha)}(q^{-2\alpha-1}x^2;q^2) ,\\
\tilde h_{2n+1,\alpha}(x;q)  &=&  (-1)^nq^{-n(2n+1)}\frac{(q;q)_{2n+1}}{(q^{2\alpha+2};q^2)_{n+1}}xL_{n}^{(\alpha+1)}(q^{-2\alpha-1}x^2;q^2).\\
\end{array}\right.
\end{equation}
 The $q$-Laguerre polynomials satisfy the following orthogonality relations (see \cite{Koekoek}):
\begin{equation}\label{orth-lag}
\ds\int_{0}^{\infty} L_{n}^{(\alpha)}(x;q^2) L_{m}^{(\alpha)}(x;q^2) x^{\alpha}e_{q^2}(-x)dx=\Gamma(-\alpha)\Gamma(\alpha+1) \frac{(q^{-2\alpha};q^2)_\infty(q^{2\alpha+2};q^2)_n}{(q^2;q^2)_\infty(q^2;q^2)_n}q^{-2n}\delta_{n,m}.
\end{equation}
Note  that the orthogonality measure in (\ref{orth-lag}) is not unique.
\begin{thm}
The $q$-polynomials $\{\tilde h_{n,\alpha}(x;q) \}_{n=0}^\infty$ satisfy the following continuous
orthogonality relations:
\begin{equation}\label{cont-orth}
\begin{array}{lll}
\ds\int_{-\infty}^{\infty} \tilde h_{n,\alpha}(x;q) \tilde h_{m,\alpha}(x;q) |x|^{2\alpha+1}\omega_\alpha(x;q)dx&=&
d_{n,\alpha}^{-2}\delta_{n,m},\\
\end{array}
\end{equation}
where
\begin{equation}\label{omega}
\omega_\alpha(x;q)=e_{q^2}(-q^{-2\alpha-1}x^2)
\end{equation}
and
\begin{equation}\label{dncoef}
d_{n,\alpha}=C_{\alpha}q^{\frac{n^2}{2}}\frac{\left(q;q\right)_{n,\alpha}^{\frac{1}{2}}}{\left(q;q\right)_{n}},\;\;\;
C_{\alpha}=\ds\sqrt{\frac{q^{-(\alpha+1)(2\alpha+1)}(q^2;q^2)_\infty }{ \Gamma(-\alpha)\Gamma(\alpha+1)(q^{-2\alpha};q^2)_\infty}}.
\end{equation}
\end{thm}
\begin{proof}
Since the weight function in the integral (\ref{cont-orth}) is an even function of the independent
variable $x$ and  the parity of the $q$-polynomials $\{\tilde h_{n,\alpha}(x;q) \}_{n=0}^\infty$ is the
parity of their degrees, it suffices to prove only those cases in (\ref{cont-orth}),
when degrees of polynomials $m$ and $n$ are either simultaneously even or odd. \\
First, we consider the even case: it follows from (\ref{exp-hn-ln}) that
\begin{equation}
\begin{array}{ll}\label{orth0}
\ds\int_{-\infty}^{\infty} \tilde h_{2n,\alpha}(x;q) \tilde h_{2m,\alpha}(x;q) |x|^{2\alpha+1}\omega_\alpha(x;q)dx&
\\=
A_{n,m}^\alpha \ds\int_{0}^{\infty}L_{n}^{(\alpha)}(q^{-2\alpha-1}x^2;q^2) L_{m}^{(\alpha)}(q^{-2\alpha-1}x^2;q^2) \omega_\alpha(x;q)x^{2\alpha+1}dx,&
\end{array}
\end{equation}
where $$A_{n,m}^\alpha =2(-1)^{n+m}q^{-n(2n-1)-m(2m-1)}\frac{(q;q)_{2n}(q;q)_{2m}}{(q^{2\alpha+2};q^2)_n(q^{2\alpha+2};q^2)_m}. $$
The change of variable $t = q^{-2\alpha-1}x^2$ in the last integral in (\ref{orth0}) leads to
$$
q^{(\alpha+1)(2\alpha+1)}\frac{A_{n,m}^\alpha}{2}  \ds\int_{0}^{\infty}L_{n}^{(\alpha)}(t;q^2) L_{m}^{(\alpha)}(t;q^2) e_{q^2}(-t)t^{\alpha}dt.
$$
By relation (\ref{orth-lag}), it follows that
\begin{equation*}
\begin{array}{ll}
\ds\int_{-\infty}^{\infty} \tilde h_{2n,\alpha}(x;q) \tilde h_{2m,\alpha}(x;q) |x|^{2\alpha+1}\omega_\alpha(x;q)dx&
\\=\Gamma(-\alpha)\Gamma(\alpha+1)
q^{(\alpha+1)(2\alpha+1)}\frac{A_{n,m}^\alpha}{2} \frac{(q^{-2\alpha};q^2)_\infty(q^{2\alpha+2};q^2)_n}{(q^2;q^2)_\infty(q^2;q^2)_{n}}q^{-2n}\delta_{n,m}\\
=\Gamma(-\alpha)\Gamma(\alpha+1) q^{(\alpha+1)(2\alpha+1)}\frac{(q^{-2\alpha};q^2)_\infty(q;q)_{2n}^2}{(q^2;q^2)_\infty(q^{2\alpha+2};q^2)_{n}(q^{2};q^2)_{n}}q^{-(2n)^2}\delta_{n,m},\\
&
\end{array}
\end{equation*}
then, using (\ref{ecr-fac}) we obtain the  result in the  case $n$ even. The odd case is obtained similarly. We have
\begin{equation}
\begin{array}{ll}\label{orth1}
\ds\int_{-\infty}^{\infty} \tilde h_{2n+1,\alpha}(x;q) \tilde h_{2m+1,\alpha}(x;q) |x|^{2\alpha+1}\omega_\alpha(x;q)dx&
\\=
B_{n,m}^\alpha \ds\int_{0}^{\infty}L_{n}^{(\alpha+1)}(q^{-2\alpha-1}x^2;q^2) L_{m}^{(\alpha+1)}(q^{-2\alpha-1}x^2;q^2)\omega_\alpha(x;q) x^{2\alpha+3}dx,&
\end{array}
\end{equation}
where
$$B_{n,m}^\alpha =2(-1)^{n+m}q^{-n(2n+1)-m(2m+1)}\frac{(q;q)_{2n+1}(q;q)_{2m+1}}{(q^{2\alpha+2};q^2)_{n+1}
(q^{2\alpha+2};q^2)_{m+1}}. $$
The change of variable $t = q^{-2\alpha-1}x^2$ in the last integral in (\ref{orth1}) leads to
$$
q^{(\alpha+2)(2\alpha+1)}\frac{B_{n,m}^{\alpha}}{2}  \ds\int_{0}^{\infty}L_{n}^{(\alpha+1)}(t;q^2)
L_{m}^{(\alpha+1)}(t;q^2) e_{q^2}(-t)t^{\alpha+1}dt.
$$
By relation (\ref{orth-lag}), it follows that
\begin{equation*}
\begin{array}{ll}
\ds\int_{-\infty}^{\infty} \tilde h_{2n+1,\alpha}(x;q) \tilde h_{2m+1,\alpha}(x;q) |x|^{2\alpha+1}\omega_\alpha(x;q)dx&
\\=\Gamma(-\alpha-1)\Gamma(\alpha+2)
q^{(\alpha+2)(2\alpha+1)}\frac{B_{n,m}^\alpha}{2} \frac{(q^{-2\alpha-2};q^2)_\infty(q^{2\alpha+4};q^2)_n}{(q^2;q^2)_\infty(q^2;q^2)_{n}}q^{-2n}\delta_{n,m}\\
=\Gamma(-\alpha-1)\Gamma(\alpha+2) q^{(\alpha+1)(2\alpha+1)}\frac{(q^{-2\alpha-2};q^2)_\infty(q^{2\alpha+4};q^2)_n(q;q)_{2n+1}^2}{(q^2;q^2)_\infty(q^{2\alpha+2};q^2)_{n+1}^2(q^{2};q^2)_{n}}q^{-4n(n+1)}\delta_{n,m}.\\
&
\end{array}
\end{equation*}
Using the fact that
$$\Gamma(-\alpha-1)\Gamma(\alpha+2)= -\Gamma(-\alpha)\Gamma(\alpha+1),$$
$$(q^{-2\alpha-2};q^2)_\infty=-q^{-2\alpha-2}(1- q^{2\alpha+2})(q^{-2\alpha};q^2)_\infty$$
 and (\ref{ecr-fac}), we get
\begin{equation*}
\begin{array}{ll}
\ds\int_{-\infty}^{\infty} \tilde h_{2n+1,\alpha}(x;q) \tilde h_{2m+1,\alpha}(x;q) |x|^{2\alpha+1}\omega_\alpha(x;q)dx&
\\
=\Gamma(-\alpha)\Gamma(\alpha+1) q^{(\alpha+1)(2\alpha+1)}\frac{(q^{-2\alpha};q^2)_\infty(q;q)_{2n+1}^2}{(q^2;q^2)_\infty(q;q)_{2n+1,\alpha}}q^{-(2n+1)^2}\delta_{n,m},
&
\end{array}
\end{equation*}
us desired.
\end{proof}

\subsection{The $q$-integral representation}
We start this section by  describing  the action of some of  $q$-derivatives operators on
powers of $x$ and on some $q$-analogues of Bessel functions.
\begin{prop}
The following statements hold:
\begin{equation}\label{derivdeltaxn}
\Delta _{q,\alpha}^kx^n=\frac{(q;q)_{n,\alpha}}{(1-q)^k(q;q)_{n-k,\alpha}}x^{n-k},\;\;n\ge k.
\end{equation}
\begin{equation}\label{dqjalpha}
D_qj_\alpha(\lambda x;q^2)=\frac{-q^2\lambda^2 x}{(1-q)(1-q^{2\alpha+2})}j_{\alpha+1}(q\lambda x;q^2).
\end{equation}
\begin{equation}\label{derivdeltaj2n}
\Delta _{q,\alpha}^{2n}j_\alpha(\lambda x;q^2)=\frac{(-1)^nq^{n(n+1)}\lambda^{2n}}{(1-q)^{2n}}j_\alpha(q^{n}\lambda x;q^2).
\end{equation}
\begin{equation}\label{derivdeltaj2n1}
\Delta _{q,\alpha}^{2n+1}j_\alpha(\lambda  x;q^2)=\frac{(-1)^{n+1}q^{(n+1)(n+2)}\lambda^{2n+2}}{(1-q)^{2n+1}(1-q^{2\alpha+2})}xj_{\alpha+1}(q^{n+1}\lambda x;q^2).
\end{equation}
\end{prop}
\begin{proof}
From  definitions (\ref{ecr-fac}), (\ref{jacksonD}) and by induction, we get (\ref{derivdeltaxn}).\\
Using the fact that
$$(q;q)_{2n,\alpha}=(1-q^{2n})(q;q)_{2n-1,\alpha}\;\;\mbox{and}\;\; (q;q)_{2n-1,\alpha}=(1-q^{2\alpha+2})(q;q)_{2n-2,\alpha+1},$$
we deduce (\ref{dqjalpha}).\\
By (\ref{jalpha}) and the result in (\ref{derivdeltaxn}) we obtain (\ref{derivdeltaj2n}).\\
By definition  (\ref{defdeltaalpha}), we have
$ \Delta _{q,\alpha}^{2n+1}j_\alpha(x;q^2)=D_q\left[\Delta _{q,\alpha}^{2n}j_\alpha\right](x;q^2)$.
Together with  (\ref{dqjalpha}) and (\ref{derivdeltaj2n}) we get (\ref{derivdeltaj2n1}).
\end{proof}
We have (see \cite[p.24, Lemma 5.1]{ghermite}).
\begin{equation}\label{expox2nalpha}
\ds\int_{0}^{\infty}e_{q^2}(-q y^2)y^{2n+2\alpha+1}d_qy  =c_{q,\alpha}q^{-n^2-2\alpha } \left(q^{2\alpha+2};q^2\right)_n,
\end{equation}
where
\begin{equation}\label{cqq}
c_{q,\alpha}=\ds\frac{(1-q)(-q^{2\alpha+3}, -q^{-2\alpha-1},q^2;q^2)_\infty}{(-q, -q,q^{2\alpha+2};q^2)_\infty}.
\end{equation}
An important formula used later is
\begin{lem}
\begin{equation}\label{q-intedjexp}
\ds\int_{0}^{\infty}e_{q^2}(-q y^2)j_\alpha(xy;q^2)y^{2\alpha+1}d_qy=c_{q,\alpha}e_{q^2}(-q^{-2\alpha-1} x^2).
\end{equation}
\end{lem}
\begin{proof}
Expand $ j_\alpha(xy;q^2)$ as power series and integrate term by term and use (\ref{expox2nalpha}) to conclude
(\ref{q-intedjexp}).
\end{proof}
 Now we provide a $q$- integral representation of generalized discrete $q$-Hermite II polynomials.
\begin{thm}
For $n=0,1,2,...,$ we have
\begin{equation}
 \begin{array}{lll}
 \tilde h_{2n,\alpha}(x;q)&=& \frac{(-1)^nq^{-n^2+n(2\alpha+3)}(q;q)_{2n}}{c_{q,\alpha}(q;q)_{2n,\alpha}
e_{q^2}(-q^{-2\alpha-1} x^2)} \\
   && \times \ds\int_{0}^{\infty}e_{q^2}(-q y^2)j_\alpha(q^{n}xy;q^2)y^{2n+2\alpha+1}d_qy.
\end{array}
\end{equation}
\begin{equation}
 \begin{array}{lll}
\tilde h_{2n+1,\alpha}(x;q) &=& \frac{(-1)^{n+1}q^{-n^2+(n+1)(2\alpha+3)}(q;q)_{2n+1}x}{c_{q,\alpha}(1-q^{2\alpha+2})(q;q)_{2n+1,\alpha}
e_{q^2}(-q^{-2\alpha-1} x^2)} \\
 &&\times\ds\int_{0}^{\infty}e_{q^2}(-q y^2)j_{\alpha+1}(q^{n+1}xy;q^2)y^{2n+2\alpha+3}d_qy.
\end{array}
\end{equation}
\end{thm}
\begin{proof}
We recall the Rodrigues-type formula (see \cite{ghermite})
\begin{equation}\label{Rodrigues-type formula}
e_{q^2}(-q^{-2\alpha-1}x^2)\tilde h_{n,\alpha}(x;q)=\frac{(q-1)^nq^{\frac{-n(n-1)}{2}}(q^{-1};q^{-1})_{n}}{(q^{-1};q^{-1})_{n,\alpha}}\Delta_{q,\alpha}^n
e_{q^2}(-q^{-2\alpha-1}x^2).
\end{equation}
From (\ref{q-intedjexp}), we obtain
\begin{equation*}
\ds\int_{0}^{\infty}e_{q^2}(-q y^2)\Delta_{q,\alpha}^nj_\alpha(xy;q^2)y^{2\alpha+1}d_qy=c_{q,\alpha}\Delta_{q,\alpha}^ne_{q^2}(-q^{-2\alpha-1} x^2).
\end{equation*}
From (\ref{derivdeltaj2n}) and (\ref{Rodrigues-type formula}), we get
\begin{equation*}
\begin{array}{ll}
\frac{(-1)^nq^{n(n+1)}}{(1-q)^{2n}}\ds\int_{0}^{\infty}e_{q^2}(-q y^2)j_\alpha(q^{n}xy;q^2)y^{2n+2\alpha+1}d_qy &\\
=c_{q,\alpha}\frac{q^{n(2n-1)}(q^{-1};q^{-1})_{2n,\alpha}}{(q-1)^{2n}(q^{-1};q^{-1})_{2n}}   e_{q^2}(-q^{-2\alpha-1} x^2)\tilde h_{2n,\alpha}(x;q).
\end{array}
\end{equation*}
Using the fact that
$\frac{(q^{-1};q^{-1})_{2n,\alpha}}{(q^{-1};q^{-1})_{2n}}=q^{-n(2\alpha+1)}\frac{(q;q)_{2n,\alpha}}{(q;q)_{2n}} $, we obtain
\begin{equation*}
\begin{array}{ll}
(-1)^n\ds\int_{0}^{\infty}e_{q^2}(-q y^2)j_\alpha(q^{n}xy;q^2)y^{2n+2\alpha+1}d_qy &\\
=c_{q,\alpha}\frac{q^{n(n-2\alpha-3)}(q;q)_{2n,\alpha}}{(q;q)_{2n}}   e_{q^2}(-q^{-2\alpha-1} x^2)\tilde h_{2n,\alpha}(x;q).
\end{array}
\end{equation*}
Hence,
\begin{eqnarray*}
 \tilde  h_{2n,\alpha}(x;q)  &=& \frac{(-1)^nq^{-n(n-2\alpha-3)}(q;q)_{2n}}{c_{q,\alpha}(q;q)_{2n,\alpha}
e_{q^2}(-q^{-2\alpha-1} x^2)} \\
   && \times \ds\int_{0}^{\infty}e_{q^2}(-q y^2)j_\alpha(q^{n}xy;q^2)y^{2n+2\alpha+1}d_qy.
\end{eqnarray*}
Further,
\begin{equation*}
\ds\int_{0}^{\infty}e_{q^2}(-q y^2)\Delta_{q,\alpha}^{2n+1}j_\alpha(xy;q^2)y^{2\alpha+1}d_qy=c_{q,\alpha}\Delta_{q,\alpha}^{2n+1}e_{q^2}(-q^{-2\alpha-1} x^2).
\end{equation*}
From (\ref{derivdeltaj2n1}) and (\ref{Rodrigues-type formula}) we get
\begin{equation*}
\begin{array}{ll}
\frac{(-1)^{n+1}q^{(n+1)(n+2)}x}{(1-q)^{2n+1}(1-q^{2\alpha+2})}\ds\int_{0}^{\infty}e_{q^2}(-q y^2)j_{\alpha+1}(q^{n+1}xy;q^2)y^{2n+2\alpha+3}d_qy &\\
=c_{q,\alpha}\frac{q^{n(2n+1)}(q^{-1};q^{-1})_{2n+1,\alpha}}{(q-1)^{2n+1}(q^{-1};q^{-1})_{2n+1}}   e_{q^2}(-q^{-2\alpha-1} x^2)\tilde h_{2n+1,\alpha}(x;q).
\end{array}
\end{equation*}
Using the fact that
$\frac{(q^{-1};q^{-1})_{2n+1,\alpha}}{(q^{-1};q^{-1})_{2n+1}}=q^{-(n+1)(2\alpha+1)}\frac{(q;q)_{2n+1,\alpha}}{(q;q)_{2n+1}} $, we obtain
\begin{equation*}
\begin{array}{ll}
\frac{(-1)^{n+1}q^{(n+1)(n+2)}x}{(1-q)^{2n+1}(1-q^{2\alpha+2})}\ds\int_{0}^{\infty}e_{q^2}(-q y^2)j_{\alpha+1}(q^{n+1}xy;q^2)
y^{2n+2\alpha+3}d_qy &\\
=c_{q,\alpha}\frac{q^{n(2n+1)-(n+1)(2\alpha+1)}(q;q)_{2n+1,\alpha}}{(q-1)^{2n+1}(q;q)_{2n+1}}   e_{q^2}(-q^{-2\alpha-1} x^2)\tilde
h_{2n+1,\alpha}(x;q).
\end{array}
\end{equation*}
Thus,
\begin{eqnarray*}
 \tilde h_{2n+1,\alpha}(x;q)&=& \frac{(-1)^{n+1}q^{-n^2+(n+1)(2\alpha+3)}(q;q)_{2n+1}x}{c_{q,\alpha}(1-q^{2\alpha+2})(q;q)_{2n+1,\alpha}
e_{q^2}(-q^{-2\alpha-1} x^2)} \\
   && \times \ds\int_{0}^{\infty}e_{q^2}(-q y^2)j_{\alpha+1}(q^{n+1}xy;q^2)y^{2n+2\alpha+3}d_qy.
\end{eqnarray*}
This completes the proof.
\end{proof}
%%%%%%%%%%%%%%%%%%%%%%%%%%%%%%%%%%%%%%%%%%%%%%%%%%%%%%%%%%%%%%%%%%%%%%%%%%%%%%%%%%%%%%%%%%%%%%%%%%%%%%%%%%%%%%%%%%%%%%%%%%%%%%%%%%%%%%%%%%%%%%%%%%%%%%%%%%%%%%
\subsection{Evaluation at unity of the Poisson kernel for $ \tilde h_{n,\alpha}(x;q)$}
\begin{thm}
The following equation is an evaluation at unity of the Poisson kernel
for the generalized discrete $q$-Hermite II polynomials:
\begin{equation}\label{summhnhn}
\begin{array}{ll}
\ds\sum_{n=0}^{\infty}\frac{q^{n^2}(q;q)_{n,\alpha}}{(q;q)_{n}^2}\tilde h_{n,\alpha}(q^{\alpha+\frac{1}{2}}x;q)
\tilde h_{n,\alpha}(q^{\alpha+\frac{1}{2}}y;q) &
\\
=\frac{(q^2;q^2)_\infty(xy)^{-\alpha}}{(q^{2\alpha+2};q^2)_\infty(x-y)}
\left[J_{\alpha+1}^{(2)}(2x;q^2)J_{\alpha}^{(2)}(2y;q^2)-J_{\alpha}^{(2)}(2x;q^2)J_{\alpha+1}^{(2)}(2y;q^2)\right].
&
\end{array}
\end{equation}
\end{thm}
\begin{proof}
From  the Christoffel$-$Darboux formula and the limit transition of
$q$-Laguerre polynomials to Jackson's $q$-Bessel functions (see \cite{Moak} and \cite{Nicola}), we
deduce that
\begin{equation*}
\begin{array}{ll}
\ds\sum_{n=0}^{\infty}\frac{q^{2n}(q^2;q^2)_{n}}{(q^{2\alpha+2};q^2)_{n}}L_n^{(\alpha)}(x^2;q^2)L_n^{(\alpha)}(y^2;q^2)  &
\\
=\frac{(q^2;q^2)_\infty(xy)^{-\alpha}}{(q^{2\alpha+2};q^2)_\infty(x^2-y^2)}
\left[xJ_{\alpha+1}^{(2)}(2x;q^2)J_{\alpha}^{(2)}(2y;q^2)-yJ_{\alpha}^{(2)}(2x;q^2)J_{\alpha+1}^{(2)}(2y;q^2)\right].
&
\end{array}
\end{equation*}
In the last sum, the $L_n^{(\alpha)}(x^2;q^2)$ can be written in terms of
$\tilde h_{2n,\alpha}(q^{\alpha+\frac{1}{2}}x;q).$  Using (\ref{exp-hn-ln}),   it follows that
\begin{equation}\label{sumh2n}
\begin{array}{ll}
\ds\sum_{n=0}^{\infty}\frac{q^{4n^2}(q;q)_{2n,\alpha}}{(q;q)_{2n}^2}\tilde h_{2n,\alpha}(q^{\alpha+\frac{1}{2}}x;q)
\tilde h_{2n,\alpha}(q^{\alpha+\frac{1}{2}}y;q) &
\\
=\ds\sum_{n=0}^{\infty}\frac{q^{2n}(q^2;q^2)_{n}}{(q^{2\alpha+2};q^2)_{n}}L_n^{(\alpha)}(x^2;q^2)L_n^{(\alpha)}(y^2;q^2)  &
\\
=\frac{(q^2;q^2)_\infty(xy)^{-\alpha}}{(q^{2\alpha+2};q^2)_\infty(x^2-y^2)}
\left[xJ_{\alpha+1}^{(2)}(2x;q^2)J_{\alpha}^{(2)}(2y;q^2)-yJ_{\alpha}^{(2)}(2x;q^2)J_{\alpha+1}^{(2)}(2y;q^2)\right].
&
\end{array}
\end{equation}
Likewise (again using (\ref{exp-hn-ln})), we have
$$\tilde h_{2n+1,\alpha}(q^{\alpha+\frac{1}{2}}x;q)=(-1)^nq^{-n(2n+1)}
\frac{(q;q)_{2n+1}}{(q^{2\alpha+2};q^2)_{n+1}} q^{\alpha+\frac{1}{2}}xL_n^{(\alpha+1)}(x^2;q^2),$$

then we get
\begin{equation*}
\begin{array}{ll}
\ds\sum_{n=0}^{\infty}\frac{q^{(2n+1)^2}(q;q)_{2n+1,\alpha}}{(q;q)_{2n+1}^2}\tilde h_{2n+1,\alpha}(q^{\alpha+\frac{1}{2}}x;q)
\tilde h_{2n+1,\alpha}(q^{\alpha+\frac{1}{2}}y;q) &
\\
=\ds\sum_{n=0}^{\infty}\frac{q^{2n+2\alpha+2}(q^2;q^2)_{n}}{(q^{2\alpha+2};q^2)_{n+1}}xyL_n^{(\alpha+1)}(x^2;q^2)
L_n^{(\alpha+1)}(y^2;q^2)  &
\\
=\frac{q^{2\alpha+2}xy}{(1-q^{2\alpha+2})}\ds\sum_{n=0}^{\infty}\frac{q^{2n}(q^2;q^2)_{n}}{(q^{2\alpha+2};q^2)_{n}}
L_n^{(\alpha+1)}(x^2;q^2)L_n^{(\alpha+1)}(y^2;q^2)  &
\\

=\frac{q^{2\alpha+2}(q^2;q^2)_\infty(xy)^{-\alpha}}{(q^{2\alpha+2};q^2)_\infty(x^2-y^2)}
\left[xJ_{\alpha+2}^{(2)}(2x;q^2)J_{\alpha+1}^{(2)}(2y;q^2)-yJ_{\alpha+1}^{(2)}(2x;q^2)J_{\alpha+2}^{(2)}(2y;q^2)\right].
&
\end{array}
\end{equation*}
Using the fact (see \cite[p.25]{G})
$$q^{2\alpha+2}xJ_{\alpha+2}^{(2)}(2x;q^2) = (1-q^{2\alpha+2})J_{\alpha+1}^{(2)}(2x;q^2)-xJ_{\alpha}^{(2)}(2x;q^2),$$
we obtain
\begin{equation}\label{sumh2n+1}
\begin{array}{ll}
\ds\sum_{n=0}^{\infty}\frac{q^{(2n+1)^2}(q;q)_{2n+1,\alpha}}{(q;q)_{2n+1}^2}\tilde h_{2n+1,\alpha}(q^{\alpha+\frac{1}{2}}x;q)
\tilde h_{2n+1,\alpha}(q^{\alpha+\frac{1}{2}}y;q) &
\\
=\frac{q^{2\alpha+2}(q^2;q^2)_\infty(xy)^{-\alpha}}{(q^{2\alpha+2};q^2)_\infty(x^2-y^2)}
\left[xJ_{\alpha+2}^{(2)}(2x;q^2)J_{\alpha+1}^{(2)}(2y;q^2)-yJ_{\alpha+1}^{(2)}(2x;q^2)J_{\alpha+2}^{(2)}(2y;q^2)\right]
&
\\
=\frac{(q^2;q^2)_\infty(xy)^{-\alpha}}{(q^{2\alpha+2};q^2)_\infty(x^2-y^2)}
\left[\left((1-q^{2\alpha+2})J_{\alpha+1}^{(2)}(2x;q^2)-xJ_{\alpha}^{(2)}(2x;q^2)\right)J_{\alpha+1}^{(2)}(2y;q^2)\right.
&
\\
 \left. -J_{\alpha+1}^{(2)}(2x;q^2)\left((1-q^{2\alpha+2})J_{\alpha+1}^{(2)}(2y;q^2)-yJ_{\alpha}^{(2)}(2y;q^2)\right)\right]
&
\\
=\frac{(q^2;q^2)_\infty(xy)^{-\alpha}}{(q^{2\alpha+2};q^2)_\infty(x^2-y^2)}
\left[ y J_{\alpha+1}^{(2)}(2x;q^2)J_{\alpha}^{(2)}(2y;q^2)-x J_{\alpha}^{(2)}(2x;q^2)J_{\alpha+1}^{(2)}(2y;q^2) \right].
&
\end{array}
\end{equation}
Add the result in (\ref{sumh2n}) to (\ref{sumh2n+1}) to obtain the desired summation.
\end{proof}
%%%%%%%%%%%%%%%%%%%%%%%%%%%%%%%%%%%%%%%%%%%%%%%%%%%%%%%%%%%%%%%%%%%%%%%%%%%%%%%%%%%%%%%%%%%%
\begin{cor}
The following  equation is  an evaluation at unity of the  Poisson kernel for
the discrete $q$-Hermite II polynomials :
\begin{equation*}
\begin{array}{ll}
\ds\sum_{n=0}^{\infty}\frac{q^{n^2}}{(q;q)_{n}}\tilde h_{n}(x;q) \tilde h_{n}(y;q) &
\\
=\frac{(q;q^2)_\infty}{(q^{2};q^2)_\infty(x-y)}
\left[Sin_q(x)Cos_q(y)-Cos_q(x)Sin_q(y)\right]
&
\end{array}
\end{equation*}
\end{cor}
\begin{proof}
For the particular case  $\alpha=-\frac{1}{2}$,
we have $\tilde h_{n,-\frac{1}{2}}(x;q)=\tilde h_{n}(x;q)  $ and
$$ J_{-\frac{1}{2}}^{(2)}(2x;q^2)=\frac{(q;q^2)_\infty}{(q^2;q^2)_\infty\sqrt{x}}Cos_q(x)\;\;\mbox{and}\;\;
J_{\frac{1}{2}}^{(2)}(2x;q^2)=\frac{(q;q^2)_\infty}{(q^2;q^2)_\infty\sqrt{x}}Sin_q(x).$$
It is esay now to finish the proof of the Corollary.
\end{proof}
Now, we express the second $q$-Bessel functions   in terms of the generalized
discrete $q$-Hermite II-polynomials:
\begin{prop}
\begin{equation}\label{summhn}
\ds\sum_{n=0}^{\infty}(-1)^n\frac{q^{n(2n+1)}(q^{2\alpha+2};q^2)_n}{(q;q)_{2n}}
\tilde h_{2n,\alpha}(q^{\alpha+\frac{1}{2}}x;q)=x^{-\alpha-1}J_{\alpha+1}^{(2)}(2x;q^2).
\end{equation}
\end{prop}
\begin{proof}
Taking limit as $y\rightarrow0$ in (\ref{summhnhn})
and using the two limits
\begin{equation}\label{limithn0}
\lim_{y\rightarrow0}y^{-\alpha}J_{\alpha}^{(2)}(2y;q^2)=\frac{(q^{2\alpha+2};q^2)_\infty}{(q^{2};q^2)_\infty},\;\;\;\lim_{y\rightarrow0}y^{-\alpha}J_{\alpha+1}^{(2)}(2y;q^2)=0,
\;\;\;\alpha\ge -\frac{1}{2},
\end{equation}
identity \ref{summhn} follows.
\end{proof}
Notice that by taking $x=0$ in (\ref{summhn})  and by appealing to the first limit in (\ref{limithn0})
and the fact that
\begin{equation*}
  \tilde h_{2n,\alpha}(0;q)=(-1)^nq^{-2n^2+n}(q;q^2)_{n}\;\;\mbox{and}\;\;
  (q;q)_{2n}=(q^2;q^2)_{n}(q;q^2)_{n},
\end{equation*}
the following special case (\cite[Theorem 10.2.1]{andrews}  with $q\mapsto q^2$, $x=q^2$ and
$a=q^{2\alpha+2}$) of the $q$-binomial theorem can be recovered:
\begin{equation}\label{summhn0}
\ds\sum_{n=0}^{\infty}\frac{q^{2n}(q^{2\alpha+2};q^2)_n}{(q^2;q^2)_{n}}=\frac{(q^{2\alpha+4};q^2)_\infty}{(q^2;q^2)_{\infty}}.
\end{equation}
%%%%%%%%%%%%%%%%%%%%%%%%%%%%%%%%%%%%%%%%%%21/03/19%%%%%%%%%%%%%%%%%%%%%%%%%%%%%%%%%%%%%%%%%%%%%
The generalized discrete $q$-Hermite II polynomials  can be  written  in terms of  basic hypergeometric
functions as:
\begin{equation}\label{hnalphad}
\left \{\begin{array}{lll}
\tilde h_{2n,\alpha}(x;q)
&=& (-1)^nq^{-n(2n-1)}(q;q^2)_{n}\ _{1}\phi _1\left(q^{-2n};q^{2\alpha+2};q^2,-q^{2n+1}x^{2} \right),\\
\tilde h_{2n+1,\alpha}(x;q)
&=&(-1)^nq^{-n(2n+1)} x\frac{(q;q^2)_{n+1}}{1-q^{2\alpha+2}}\ _{1}\phi _1\left(q^{-2n};q^{2\alpha+4};q^2,-q^{2n+3}x^{2}
\right).
\end{array}\right.
\end{equation}
\begin{thm}\label{csqgenfun}
\begin{equation}\label{newformula}
e_{q^2}(-z^2)\ds\sum_{n=0}^\infty\frac{(-1)^nq^{2n(n+\alpha)}z^{2n}}{(q^2;q^2)_n(q^{2\alpha+2};q^{2})_n}=
  \ds\sum_{n=0}^\infty\frac{(-1)^nz^{2n}}{(q^2;q^2)_n(q^{2\alpha+2};q^{2})_n},\;\; |z|<1.
\end{equation}
\end{thm}
\begin{proof}
From (\ref{hnalphad}), we have
\begin{equation*}
\tilde h_{2n,\alpha}(iq^{\alpha+1/2};q)= (-1)^nq^{-n(2n-1)}(q;q^2)_{n}\ _{1}\phi _1\left(q^{-2n};q^{2\alpha+2};q^2,q^{2n+2\alpha+2}\right).
\end{equation*}
By the summation formula of $\ _{1}\phi _1$ series (see \cite[Appendix II, (II.5)]{G}):
\begin{equation*}
\ _{1}\phi _1\left(a;c;q,c/a\right) =\frac{(c/a;q)_\infty}{ (c;q)_\infty},
\end{equation*}
we get
\begin{equation*}
\tilde h_{2n,\alpha}(iq^{\alpha+1/2};q)=  (-1)^nq^{-n(2n-1)}(q;q^2)_{n}\frac{(q^{2n+2\alpha+2};q^{2})_\infty}{(q^{2\alpha+2};q^{2})_\infty},
\end{equation*}
which can be written as
\begin{equation}\label{h2nfacto}
\tilde h_{2n,\alpha}(iq^{\alpha+1/2};q)=  (-1)^nq^{-n(2n-1)}\frac{(q;q^2)_{n}}{(q^{2\alpha+2};q^{2})_n}.
\end{equation}
In particular, setting $x = iq^{\alpha+1/2}$  in the even part of the generating function (\ref{gghntild})
and using (\ref{h2nfacto}), it follows that
\begin{equation*}
  e_{q^2}(-z^2)\ds\sum_{n=0}^\infty\frac{(-1)^nq^{n(2n-1)}q^{n(2\alpha+1)}}{(q^2;q^2)_n(q^{2\alpha+2};q^{2})_n}z^{2n}=
  \ds\sum_{n=0}^\infty\frac{(-1)^n}{(q^2;q^2)_n(q^{2\alpha+2};q^{2})_n}z^{2n}.
\end{equation*}
This identity holds for $z$ in the open unit disk.
\end{proof}
%%%%%%%%%%%%%%%%%%%%%%%%%%%%%%%%%%%%%%%%%%%%%%%%%%%%%%%%%%%    15/04/19   %%%%%%%%%%%%%%%%%%%%%%%%%%%
\begin{rem}
Note that  formula (\ref{newformula}) can be deduced from the  Heine transformations  of $\ _{2}\phi_1$ series   (see \cite[Appendix III, (III.3)]{G}):
\begin{equation}\label{heine}
\ _{2}\phi_1(a,b;c;q,z)=\frac{(abz/c;q)_\infty}{(z;q)_\infty}\ _{2}\phi_1(c/a,c/b;c;q,abz/c).
\end{equation}
If one replaces $q$ by $q^2$, $c$ by $q^{2\alpha+2}$, $z$ by $-z^2 q^{2\alpha+2}/ab$, respectively in (\ref{heine}),
and then send $a\to\infty$ and $b\to\infty$, one obtains the transformation (\ref{newformula}).
\end{rem}
%%%%%%%%%%%%%%%%%%%%%%%%%%%%%%%%%%%%%%%%%%%%%%%%%%%%%%%%%%%%%%%%%%%%%%%%%%%%%%%%%%%%%%%%%%%%%%%%
Notice that,  setting   $x = iq^{\alpha+1/2}$ in the odd part of the generating function
(\ref{gghntild}), and using the same method us above,
one can derive an equivalent formula to (\ref{newformula}) in which $\alpha$ is
replaced by $\alpha + 1.$\\
In the next corollary,  we  recover  the two  well-known  Ramanujan's identities \cite[Chapeter I, p. 33$-$34, identities (1.7.14) and (1.7.17)]{andrews2}:
%, and  is a Ramanujan's result
%(see \cite[Chapter I, p. 33$-$34, identities (1.7.14) and (1.7.17)]{andrews2}).
\begin{cor}
\begin{equation}\label{form1}
\ds\sum_{n=0}^\infty\frac{q^{n}}{(q;q)_{2n}}=\frac{(-q^3;q^8)_\infty(-q^5;q^8)_\infty(q^8;q^8)_\infty}{(q;q)_\infty},
\end{equation}
\begin{equation}\label{form2}
\ds\sum_{n=0}^\infty\frac{q^{n}}{(q;q)_{2n+1}}=\frac{(-q;q^8)_\infty(-q^7;q^8)_\infty(q^8;q^8)_\infty}{(q;q)_\infty}.
\end{equation}
\end{cor}
\begin{proof}
Setting $\alpha=-\frac{1}{2}$ and $z=iq^{1/2}$ in (\ref{newformula}), and using the identity (39) in
Slater's list \cite{Slater}:
\begin{equation*}
\ds\sum_{n=0}^\infty\frac{q^{2n^2}}{(q;q)_{2n}}=\frac{(-q^3;q^8)_\infty(-q^5;q^8)_\infty(q^8;q^8)_\infty}{(q^2;q^2)_\infty},
\end{equation*}
and the fact that $(q;q)_\infty=(q;q^2)_\infty (q^2;q^2)_\infty$, it follows formula (\ref{form1}).\\
Multiplying both sides of (\ref{newformula}) by  $\frac{1}{1-q}$ and  setting  $\alpha=\frac{1}{2}$ and
$z=iq^{1/2}$ in the resulting equation, and then using the identity (38) in Slater's list \cite{Slater}:
\begin{equation*}
\ds\sum_{n=0}^\infty\frac{q^{2n(n+1)}}{(q;q)_{2n+1}}=\frac{(-q;q^8)_\infty(-q^7;q^8)_\infty(q^8;q^8)_\infty}{(q^2;q^2)_\infty},
\end{equation*}
it follows the summation formula (\ref{form2}).
\end{proof}
\begin{rem}
  The equivalence of identities (38) and (39) in Slater's list \cite{Slater}  with the two corresponding
  identities (\ref{form2}) and (\ref{form1}) due to Ramanujan is certainly known,
as this corresponds to a special case of  of the Heine transformation formula (\ref{heine}).
\end{rem}
\begin{prop} For $|qz|<1$, we have
\begin{equation}\label{newform2}
e_{q^2}(-q^2z^2)\ds\sum_{n=0}^\infty\frac{(-1)^n q^{2n(n + \alpha)}}{(q^2;q^2)_n(q^{2\alpha+2};q^{2})_n}z^{2n}=
\ds\sum_{n=0}^\infty\frac{(-1)^nq^{2n}\left(1+q^{2\alpha}-q^{2n+2\alpha} \right)  }{(q^2;q^2)_n(q^{2\alpha+2};q^{2})_n}z^{2n}.
\end{equation}
\end{prop}
\begin{proof}
Multiply both sides of (\ref{newformula}) by $1+z^2,$ it follows that
\begin{eqnarray*}
e_{q^2}(-q^2z^2)\ds\sum_{n=0}^\infty\frac{(-1)^n q^{2n(n+ \alpha)}}{(q^2;q^2)_n(q^{2\alpha+2};q^{2})_n}z^{2n} = (1+z^2) \ds\sum_{n=0}^\infty\frac{(-1)^n}{(q^2;q^2)_n(q^{2\alpha+2};q^{2})_n}z^{2n}.\\
  =  1+ \ds\sum_{n=1}^\infty\frac{(-1)^n}{(q^2;q^2)_n(q^{2\alpha+2};q^{2})_n}z^{2n} - \ds\sum_{n=1}^\infty\frac{(-1)^n}{(q^2;q^2)_{n-1}(q^{2\alpha+2};q^{2})_{n-1}}z^{2n} \\
  =  1+ \ds\sum_{n=1}^\infty\frac{(-1)^n\left[1-(1-q^{2n})(1-q^{2n+2\alpha}) \right] }{(q^2;q^2)_n(q^{2\alpha+2};q^{2})_n}z^{2n} \\
 =  \ds\sum_{n=0}^\infty\frac{(-1)^nq^{2n}\left(1+q^{2\alpha}-q^{2n+2\alpha} \right)  }{(q^2;q^2)_n(q^{2\alpha+2};q^{2})_n}z^{2n},
\end{eqnarray*}
which completes the proof.
\end{proof}
\begin{cor}
\begin{equation}\label{formsum}
 \frac{1}{(q^2;q^2)_\infty(q^{2\alpha+2};q^{2})_\infty}=\ds\sum_{n=0}^\infty\frac{q^{2n}\left(1+q^{2\alpha}-q^{2n+2\alpha}
\right)  }{(q^2;q^2)_n(q^{2\alpha+2};q^{2})_n}.
\end{equation}
\end{cor}
\begin{proof}
Setting $z=i$ in (\ref{newform2}), we obtain
\begin{equation*}
e_{q^2}(q^2)\ds\sum_{n=0}^\infty\frac{q^{2n(n-1)}q^{n(2\alpha+2)}}{(q^2;q^2)_n(q^{2\alpha+2};q^{2})_n}=
\ds\sum_{n=0}^\infty\frac{q^{2n}\left(1+q^{2\alpha}-q^{2n+2\alpha} \right)  }{(q^2;q^2)_n(q^{2\alpha+2};q^{2})_n}.
\end{equation*}
Applying the following Cauchy's formula (see \cite{andrews}, p. 522)
\begin{equation*}
 \ds\sum_{n=0}^\infty\frac{q^{n(n-1)}x^n }{(q;q)_n(x;q)_n}=\frac{1}{(x;q)_\infty},
\end{equation*}
it follows the identity (\ref{formsum}).
\end{proof}
%%%%%%%%%%%%%%%%%%%%%%%%%%%%%%%%%%%%%%%%%%%%%%%%%%%%%%%%%%%%%%%%%%5
\section{Realization of the  quantum algebra $\mathsf{su}_{q^{\frac{1}{2}}}(1, 1)$ }
The quantum algebra $ \mathsf{su}_{q}(1, 1)$ is defined as the associative unital  algebra generated
by the operators $\{ K_-,\;K_+,\; K_0\}$ which satisfy
the conjugation relations (see \cite{Kulish})
$$(K_0)^*=K_0,\;\;(K_+)^*=K_-, $$
and the  commutation relations
$$[K_0,K_{\pm}]=\pm K_{\pm},\;\;\;\;[K_{-}, K_{+}]= \left[ 2K_0 \right]_{q^2},$$
where $[x]_q={{q^{x}-q^{-x}}\over{q^{}-q^{-1}}}$ is  a symmetric definition of $q$-numbers, invariant
by $q \leftrightarrow q^{-1}$.\\
The Casimir operator $C$, which by definition commutes with the generators $K_{\pm} $ and $K_0$     is
 $$C= \left[ K_0-\frac{1}{2}\right]_{q^2}^2-K_{+}K_{-}.    $$
Now,  we discuss an explicit  one-dimensional realization of the  quantum
algebra $\mathsf{su}_{q^{\frac{1}{2}}}(1, 1)$.
We give  a  concrete functional realization of the Hilbert space $\mathfrak{H}$ (defined just
below) and   an explicit expression of the representation
operators $K_-$, $K_+$ and $K_0$ defined in preceding paragraph in terms of
 $q$-difference  operators.\\
For this purpose, first we take   $\mathfrak{H}=L_{\alpha}^2(\R) $ to be the  space of
functions $\psi(x)$ such that
\begin{equation*}
\int_{-\infty}^{\infty}|\psi(x)|^2|x|^{2 \alpha+1}dx<\infty
\end{equation*}
with the scalar product
\begin{equation*}
\langle\psi_1,\psi_2\rangle=\int_{-\infty}^{\infty}\psi_1(x)\overline{\psi_2(x)}|x|^{2\alpha+1}dx.
\end{equation*}
Now, we construct a convenient orthonormal basis  of  $ L_{\alpha}^2(\R)$ consisting of
$(q,\alpha)$-deformed  Hermite  functions defined by
\begin{equation}\label{wave}
\phi_{n}^\alpha(x;q)=d_{n,\alpha}\sqrt{\omega_\alpha(x;q)}\tilde h_{n,\alpha}(x;q),
\end{equation}
where $\tilde h_{n,\alpha}(x;q)$, $\omega_\alpha(x;q)$ and $d_{n,\alpha}$ are given by (\ref{hnalphatild}),
(\ref{omega}) and (\ref{dncoef}),  respectively.

\begin{prop}
$\{  \phi_{n}^\alpha(x;q)   \}_{n=0}^\infty $  is a complete
orthonormal set in  $ L_{\alpha}^2(\R)$.
\end{prop}
\noindent \textbf{Proof:}\\
The  continuous orthogonality relation (\ref{orthohnalpht}) for $\tilde h_{n,\alpha}(x;q)$
can be written as
\begin{equation*}
\ds\int_{-\infty}^{\infty}\phi_{n}^\alpha(x;q)\phi_{m}^\alpha(x;q) |x|^{2\alpha+1}dx= \delta_{n,m}.
\end{equation*}
Thus $\{ \phi_{n}^\alpha(x;q) \}_{n=0}^\infty $ is an orthonormal set in  $L_{\alpha}^2(\R)$.
We prove that it is a complete set.
Suppose that there exists $f\in L_{\alpha}^2(\R)$ orthogonal to all $ \phi_{n}^\alpha(x;q)$,
$$ \int_{-\infty}^{\infty} \phi_{n}^\alpha(x;q) f(x)|x|^{2\alpha +1}dx=0,\;\;\;\mbox{for all}\;\; n\in \N .$$
By using the inverse formula (\ref{monometild}), we obtain
$$ \int_{-\infty}^{\infty}\sqrt{\omega_\alpha(x;q)} x^n  f(x)|x|^{2\alpha +1}dx=0,\;\;\;\mbox{for all}\;\;
n\in \N .$$
Using the technique that appears in \cite{Akhiezer} (p. 26),
we deduce that $f=0$. \ef
Let $\delta_{q}$ be the $q$-dilatation operator in the variable $x$, i.e.
$\delta_{q}f(x)=f(qx).$  The operator of multiplication by a
function $g$ will be denoted also by $g$.\\
Let $\mathfrak{S}_{q\alpha}$ be the  finite linear  span of $(q,\alpha)$-deformed  Hermite
functions $\phi_{n}^\alpha(x;q)$.\\
%%%%%%%%%%%%%%%%%%%%%%%%%%%%%%%%%%%%%%%%%%%%%%%%%%%%%%%%%%%%%%%%%%%%%%%%%%%%%%%%%%%%%%%%%%%%%%%%%%%%%%%%%%
It is well known that a solution of the stationary Schr\"{o}dinger equation is represented by
eigenfunctions of the Schr\"{o}dinger operator.
\begin{defn} The $q$-Schr\"{o}dinger operator $H$ acting on any function $f$ in $ L_{\alpha}^2(\R)$  is
defined by
\begin{equation}\label{oper-H}
Hf=
\begin{pmatrix}
H_e & 0 \\
0 &  H_o
\end{pmatrix}
\begin{pmatrix}
f_e \\
f_o
\end{pmatrix}
\end{equation}
where
\begin{equation*}
\begin{array}{lll}
H_e & = & -\frac{q^{2\alpha+1}}{(1-q)x^2}  \\
 && \times  \left[q^{-2\alpha}\delta_{q^{-1}}\sqrt{1+q^{-2\alpha-1}x^2}+\sqrt{1+q^{-2\alpha-1}x^2}\delta_{q}
-(1+q^{-2\alpha}+q^{-2\alpha-1}x^2)I \right],
\end{array}
\end{equation*}

\begin{equation*}
\begin{array}{lll}
H_o &  =  & -\frac{q^{2\alpha+1}}{(1-q)x^2} \\ && \times  \left[q\delta_{q^{-1}}\sqrt{1+q^{-2\alpha-1}x^2}+
q^{2\alpha+1} \sqrt{1+q^{-2\alpha-1}x^2}\delta_{q}-(1+q^{2\alpha+2}+q^{-2\alpha-1}x^2)I \right],
\end{array}
\end{equation*}
 $f_e$ and $f_o$ are respectively the even and the odd parts of $f$ and $I$ is the identity operator.
\end{defn}
\begin{thm}
$H $ is a self-adjoint operator in $\mathfrak{S}_{q\alpha}$,  with eigenfunctions
\begin{equation*}
 \phi_{n}^\alpha(x;q),\;\;n=0,1,2,...,
\end{equation*}
and we have
 \begin{equation*}
H\phi_{n}^\alpha(x;q) = \left\llbracket n \right\rrbracket_{q,\alpha}\phi_{n}^\alpha(x;q).
 \end{equation*}
\end{thm}
\begin{proof}
Let $f,g \in \mathfrak{S}_{q\alpha} $, $f=f_e+f_o, \;\;\; g=g_e+g_o .$
Due to the parity of the integrand in $(Hf,g)$, we can write
$$(Hf,g)=(H_ef_e,g_e) +(H_of_o,g_o), $$
where
\begin{eqnarray*}
(H_ef_e,g_e)&=&
-\frac{q^{2\alpha+1}}{(1-q)}\int_{-\infty}^{\infty} \frac{q^{-2\alpha}\sqrt{1+q^{-2\alpha-3}x^2}}{x^2}f_e(q^{-1}x)\overline{g_e(x)}|x|^{2\alpha+1}dx\\
&&-\frac{q^{2\alpha+1}}{(1-q)}\int_{-\infty}^{\infty}\frac{\sqrt{1+q^{-2\alpha-1}x^2}}{x^2}f_e(qx)\overline{g_e(x)}|x|^{2\alpha+1}dx\\
&&+\frac{q^{2\alpha+1}}{(1-q)}\int_{-\infty}^{\infty}\frac{ (1+q^{-2\alpha}+q^{-2\alpha-1}x^2)}{x^2} f_e(x)\overline{g_e(x)}|x|^{2\alpha+1}dx,
\end{eqnarray*}
\begin{eqnarray*}
(H_of_o,g_o)&=&
-\frac{q^{2\alpha+1}}{(1-q)}\int_{-\infty}^{\infty} \frac{q\sqrt{1+q^{-2\alpha-3}x^2}}{x^2}f_o(q^{-1}x)\overline{g_o(x)}|x|^{2\alpha+1}dx\\
&&-\frac{q^{2\alpha+1}}{(1-q)}\int_{-\infty}^{\infty}\frac{q^{2\alpha+1}\sqrt{1+q^{-2\alpha-1}x^2}}{x^2}f_o(qx)\overline{g_o(x)}|x|^{2\alpha+1}dx\\
&&+\frac{q^{2\alpha+1}}{(1-q)}\int_{-\infty}^{\infty}\frac{ (1+q^{2\alpha+2}+q^{-2\alpha-1}x^2)}{x^2} f_o(x)\overline{g_o(x)}|x|^{2\alpha+1}dx.
\end{eqnarray*}
Using the substitutions $u=q^{-1}x$ in the first integral and $u=qx$ in the second integral, we obtain
\begin{eqnarray*}
(H_ef_e,g_e)&=&
-\frac{q^{2\alpha+1}}{(1-q)}\int_{-\infty}^{\infty}f_e(u) \overline{\frac{\sqrt{1+q^{-2\alpha-1}u^2}}{u^2}g_e(qu)}|u|^{2\alpha+1}du\\
&&-\frac{q^{2\alpha+1}}{(1-q)}\int_{-\infty}^{\infty}f_e(u)\overline{\frac{q^{-2\alpha}\sqrt{1+q^{-2\alpha-3}u^2}}{u^2}g_e(q^{-1}u)}|u|^{2\alpha+1}du\\
&&+\frac{q^{2\alpha+1}}{(1-q)}\int_{-\infty}^{\infty}f_e(u)\overline{\frac{ (1+q^{-2\alpha}+q^{-2\alpha-1}u^2)}{u^2} g_e(u)}|u|^{2\alpha+1}du\\
&=&(f_e,H_eg_e).
\end{eqnarray*}
The same argument can prove that $(H_of_o,g_o)=(f_o,H_og_o)$. Therefore, we conclude that,
$H$ is a self-adjoint operator in $\mathfrak{S}_{q\alpha}$.\\
We have
\begin{eqnarray*}
H_e(\phi_{2n}^\alpha(x;q))&=& d_{2n,\alpha}H_e\left[\sqrt{\omega_\alpha(x;q)}\tilde h_{2n,\alpha}(q^{-1}x;q) \right]\\
&=& -\frac{q^{2\alpha+1}}{(1-q)x^2}d_{2n,\alpha}\sqrt{\omega_\alpha(x;q)}\left[q^{-2\alpha}\tilde h_{2n,\alpha}(q^{-1}x;q)+ (1+q^{-2\alpha-1}x^2)\tilde h_{2n,\alpha}(qx;q)\right.\\
&&\left.-(1+q^{-2\alpha}+q^{-2\alpha-1}x^2) \tilde h_{2n,\alpha}(x;q) \right].
\end{eqnarray*}
Using the relation (\ref{h2n-q-diff}), it follows that
\begin{eqnarray*}
H_e(\phi_{2n}^\alpha(x;q))&=& -\frac{q^{2\alpha+1}}{(1-q)x^2}d_{2n,\alpha}\sqrt{\omega_\alpha(x;q)}\left[(q^{2n-2\alpha-1}-q^{-2\alpha-1})x^2\tilde h_{2n,\alpha}(x;q) \right]\\
&=&\left\llbracket 2n \right\rrbracket_{q,\alpha}\phi_{2n}^\alpha(x;q),
\end{eqnarray*}
and
\begin{equation*}
\begin{array}{lll}
H_o(\phi_{2n+1}^\alpha(x;q))= d_{2n+1,\alpha}H_o\left[\sqrt{\omega_\alpha(x;q)}\tilde h_{2n+1,\alpha}(q^{-1}x;q)
\right] \\
\hspace{2.7cm} =-\frac{q^{2\alpha+1}}{(1-q)x^2}d_{2n+1,\alpha}\sqrt{\omega_\alpha(x;q)}  \\
\times \left[q\tilde h_{2n+1,\alpha}(q^{-1}x;q)+ q^{2\alpha+1}(1+q^{-2\alpha-1}x^2)\tilde h_{2n,\alpha}(qx;q)
-(1+q^{2\alpha+2}+q^{-2\alpha-1}x^2) \tilde h_{2n,\alpha}(x;q)
\right]  \\ \hspace{2.7cm}
= -\frac{q^{4\alpha+2}}{(1-q)x^2}d_{2n+1,\alpha}\sqrt{\omega_\alpha(x;q)}   \\
\times \left[q^{-2\alpha}\tilde
h_{2n+1,\alpha}(q^{-1}x;q)+ (1+q^{-2\alpha-1}x^2)\tilde h_{2n+1,\alpha}(qx;q)
-(q+q^{-2\alpha-1}+q^{-4\alpha-2}x^2) \tilde h_{2n+1,\alpha}(x;q) \right].
\end{array}
\end{equation*}
Using the relation (\ref{h2n1-q-diff}), it follows that
\begin{eqnarray*}
H_o(\phi_{2n+1}^\alpha(x;q))&=& -\frac{q^{4\alpha+2}}{(1-q)x^2}d_{2n+1,\alpha}\sqrt{\omega_\alpha(x;q)}\left[(q^{2n-2\alpha}-q^{-4\alpha-2})x^2\tilde h_{2n+1,\alpha}(x;q) \right]\\
&=&\left\llbracket 2n+1 \right\rrbracket_{q,\alpha}\phi_{2n+1}^\alpha(x;q),
\end{eqnarray*}
us desired.
\end{proof}
Let us note that, due to the regularity of $ \phi_{n}^\alpha(x;q)$, the singularity of
$H$ at $x=0$ can be omitted when we apply $H$ to the function $f\in\mathfrak{S}_{q\alpha} $.\\
%%%%%%%%%%%%%%%%%%%%%%%%%%%%%%%%%%%%%%%%%%%%%%%%%%%%%%%%%%%%%%%%%%%%%%%%%%%%%%%%%%%%%%%%%%%%%%%%%%%%%%
%%%%%%%%%%%%
From  the forward and backward shift operators (\ref{forwardshifttild}) and  (\ref{backward}),
we define the  operators $a$ and $a^+$  on  $\mathfrak{S}_{q\alpha}$   by means of
$2\times 2$ matrix forms:
\begin{equation}\label{oper-a}
af=\frac{q^\frac{1}{2}}{\sqrt{1-q}x}
\begin{pmatrix}
\delta_{q^{-1}}\sqrt{1+q^{-2\alpha-1}x^2}-1 & 0 \\
0 &  \delta_{q^{-1}}\sqrt{1+q^{-2\alpha-1}x^2}- q^{2\alpha+1}
\end{pmatrix}%
\begin{pmatrix}
f_e \\
f_o
\end{pmatrix},%
\end{equation}
\begin{equation}\label{oper-a+}
a^{+}f=\frac{q^{2\alpha+\frac{3}{2}}}{\sqrt{1-q}x}
\begin{pmatrix}
\sqrt{1+q^{-2\alpha-1}x^2}\delta_{q}-1 & 0 \\
0 &  \sqrt{1+q^{-2\alpha-1}x^2}\delta_{q}- q^{-2\alpha-1}
\end{pmatrix}%
\begin{pmatrix}
f_e \\
f_o
\end{pmatrix}.%
\end{equation}

The reader may  verify that these operators are indeed mutually adjoint
in the Hilbert space $ L_{\alpha}^2(\R)$.\\
So,  the $q$-Schr\"{o}dinger operator $H$ can be factorized as
$$H=a^{+}a. $$
The  action of the operators  $a$ and $a^+$ on the
basis   $\{\phi_{n}^\alpha(x;q)\}_{n=0}^\infty $ of   $ L_{\alpha}^2(\R)$ leads to the explicit results:
\begin{prop}
The following statements hold:
\begin{eqnarray}
a \phi_{0}^\alpha(x;q) &=&0,\label{aphi0}\\
a \phi_{n}^\alpha(x;q) &=&\sqrt{\left \llbracket n \right\rrbracket_{q,\alpha}}\phi_{n-1}^\alpha(x;q),\;\;\;n\ge1,\label{aphin}\\
a^+\phi_{n}^\alpha(x;q) &=&\sqrt{\left\llbracket n+1\right\rrbracket_{q,\alpha}}\phi_{n+1}^\alpha(x;q),\label{a+phin}\\
\phi_{n}^\alpha(x;q)&=& (n!_{q,\alpha})^{-\frac{1}{2}}a^{+n} \phi_{0}^\alpha(x;q)\label{phinexp},
\end{eqnarray}
where  $\left\llbracket n \right\rrbracket_{q,\alpha}$ is defined by (\ref{nfactalpha}).
\end{prop}
\noindent \textbf{Proof:}\\
Formula (\ref{aphin}) follows from the forward and backward shift operators (\ref{forwardshifttild}) and  (\ref{backward}) and from the fact that
\begin{equation*}
d_{n,\alpha}=\frac{q^{n-\frac{1}{2}}\sqrt{\left\llbracket n \right\rrbracket_{q,\alpha}}}{\sqrt{1-q}\left\llbracket n\right\rrbracket_q}d_{n-1,\alpha}.
 \end{equation*}
Formula (\ref{aphi0}) is an immediate consequence of the definition (\ref{wave}) and
(\ref{aphin}). Finally
(\ref{phinexp}) is a consequence of (\ref{a+phin}).
\ef
From  (\ref{aphin}) and (\ref{a+phin})  one deduces that
\begin{equation}\label{eqaa+}
a^+a\phi_{n}^{\alpha}(x;q) =\left\llbracket n\right\rrbracket_{q,\alpha}\phi_{n}^{\alpha}(x;q),
\end{equation}
\begin{equation}\label{eqa+a}
 aa^+\phi_{n}^{\alpha}(x;q) =\left\llbracket n+1\right\rrbracket_{q,\alpha}\phi_{n}^{\alpha}(x;q).
\end{equation}
The number operator $N$ is defined in this case by the relations
\begin{equation}\label{op-numb}
a^+ a=\left\llbracket N\right \rrbracket_{q,\alpha},\;\;\;aa^+=\left\llbracket N+1\right\rrbracket_{q,\alpha}\;\;\;\mbox{on}\;\;\mathfrak{S}_{q\alpha}.
\end{equation}
The formulas  (\ref{op-numb}) can be inverted to  determine an explicit expression  of the  operator $N$  as follows:
\begin{equation}\label{particule-operator}
N:=\frac{1}{2\log q}\log\left(1-(1-q)aa^+\right)+\frac{1}{2\log q}\log\left(1-(1-q)a^+a\right)-\alpha-1.
\end{equation}
From (\ref{eqaa+}), (\ref{eqa+a}) and  (\ref{particule-operator}), we obtain:
\begin{equation}
 N\phi_{n}^{\alpha}(x;q)=n\phi_{n}^{\alpha}(x;q),
\end{equation}
and
\begin{equation}
[N,a]=-a,\;\;\; [N,a^+]=a^+\;\;\;\mbox{on}\;\;\mathfrak{S}_{q\alpha}.
\end{equation}
Now, we consider the operators
$$ b=q^{-\frac{N+(K+1)(\alpha+\frac{1}{2})}{4}} a,\;\;\;    b^+=a^+ q^{-\frac{N+(K+1)(\alpha+\frac{1}{2})}{4}},\;\; K=(-1)^N.$$
Using  the relation
\begin {equation*}\label{relatqnumb}
[x]_{q^\frac{1}{2}}=q^{-\frac{x-1}{2}}\llbracket x\rrbracket_{q},
\end {equation*}
one easily verifies that the actions of the operators $b$ and  $b^+$ on the basis
$\{\phi_{n}^\nu(x;q)\}_{n=0}^\infty  $ are given by
\begin{equation}\label{bb+actions}
\begin{array}{llll}
 b\phi_{2n}^\alpha(x;q) &=& \sqrt{\left[ 2n\right]_{q^{\frac{1}{2}}}}\phi_{2n-1}^\alpha(x;q)  ,&n\ge1,\\
  b\phi_{2n+1}^\alpha(x;q) &=& \sqrt{\left[  2n+2\alpha+2\right]_{q^{\frac{1}{2}}}}\phi_{2n}^\alpha(x;q),\\
b^+\phi_{2n}^\alpha(x;q)&=& \sqrt{\left[2n+2\alpha+2\right]_{q^{\frac{1}{2}}}} \phi_{2n+1}^\alpha(x;q),\\
 b^+\phi_{2n+1}^\nu(x;q) &=& \sqrt{\left[ 2n+2\right]_{q^{\frac{1}{2}}}}\phi_{2n+2}^\alpha(x;q) .\\
\end{array}
\end{equation}
Now we are ready to construct an explicit realization of the operators $K_-$, $K_+$ and $K_0$
generators of the quantum algebra $su_{q^{\frac{1}{2}}}(1,1)$ in terms of
the oscillatorial operators $a$, $a^+$ and $N$ by setting,
$$ K_-=\gamma \left(b\right)^2,\;\;\;K_+=\gamma \left(b^+\right)^2,\;\;\;K_0=\frac{1}{2}(N+\alpha+1),\;\;\;\gamma=(\left[  2\right]_{q^{\frac{1}{2}}})^{-1}.$$
From (\ref{bb+actions}) we derive  the actions of these operators on the  basis $\{  \phi_{n}^\nu(x;q)   \}_{n=0}^\infty $:
\begin{equation}\label{B0B+B}
\begin{array}{llll}
K_0\phi_{n}^\alpha(x;q)&=& \frac{1}{2}(n+\alpha+1)\phi_{n}^\alpha(x;q),\\
K_+\phi_{2n}^\alpha(x;q)&=& \gamma\sqrt{\left[2n+2\right]_{q^{\frac{1}{2}}} \left[2n+2\alpha+2\right]_{q^{\frac{1}{2}}}}\phi_{2n+2}^\alpha(x;q),\\
K_+\phi_{2n+1}^\alpha(x;q)&=& \gamma\sqrt{\left[ 2n+2\right]_{q^{\frac{1}{2}}}\left[ 2n+2\alpha+4\right]_{q^{\frac{1}{2}}}}\phi_{2n+3}^\alpha(x;q),\\
K_-\phi_{2n}^\alpha(x;q) &=& \gamma\sqrt{\left[ 2n\right]_{q^{\frac{1}{2}}}\left[ 2n+2\alpha\right]_{q^{\frac{1}{2}}}}\phi_{n-2}^\alpha(x;q),\;\;n\ge1,\\
K_- \phi_{2n+1}^\alpha(x;q)&=& \gamma\sqrt{\left[2n\right]_{q^{\frac{1}{2}}}\left[  2n+2\alpha+2\right]_{q^{\frac{1}{2}}}}\phi_{2n-1}^\alpha(x;q),\;\;n\ge1.\\
\end{array}
\end{equation}
It follows that
\begin{equation}\label{BB+act}
\begin{array}{llll}
K_-K_+\phi_{2n}^\alpha(x;q)&=& \gamma^2\left[ 2n+2\right]_{q^{\frac{1}{2}}}  \left[ 2n+2\alpha+2\right]_{q^{\frac{1}{2}}}\phi_{2n}^\alpha(x;q),\\
K_-K_+\phi_{2n+1}^\alpha(x;q)&=& \gamma^2\left[2n+2\right]_{q^{\frac{1}{2}}}\left[ 2n+2\alpha+4\right]_{q^{\frac{1}{2}}}\phi_{2n+1}^\alpha(x;q),\\
K_+K_-\phi_{2n}^\alpha(x;q) &=& \gamma^2\left[2n\right]_{q^{\frac{1}{2}}}\left[ 2n+2\alpha\right]_{q^{\frac{1}{2}}}\phi_{2n}^\alpha(x;q),\\
K_+K_- \phi_{2n+1}^\alpha(x;q)&=& \gamma^2\left[2n\right]_{q^{\frac{1}{2}}}\left[  2n+2\alpha+2\right]_{q^{\frac{1}{2}}}\phi_{2n+1}^\alpha(x;q).\\
\end{array}
\end{equation}
Using the following identity (see \cite{Biedenharn} p.58)
\begin{equation}\label{q-numb-add}
 \left[ a \right]_{q}\left[ b-c \right]_{q}
+\left[ b \right]_{q}\left[ c-a \right]_{q}
+\left[ c \right]_{q}\left[ a-b \right]_{q}=0,
\end{equation}
with $a=2n+2$, $b=-2n-2\alpha$, $c=2$ and $a=2n+2$, $b=-2n-2\alpha-2$, $c=2$ respectively,  we obtain
$$\left[2n+2\right]_{q^{\frac{1}{2}}}  \left[ 2n+2\alpha+2\right]_{q^{\frac{1}{2}}}-\left[ 2n\right]_{q^{\frac{1}{2}}}\left[  2n+2\alpha\right]_{q^{\frac{1}{2}}}=
\left[ 2\right]_{q^{\frac{1}{2}}}
\left[  4n+2\alpha+2\right]_{q^{\frac{1}{2}}},
 $$
$$
\left[ 2n+2\right]_{q^{\frac{1}{2}}}\left[ 2n+2\alpha+4\right]_{q^{\frac{1}{2}}}-
\left[ 2n\right]_{q^{\frac{1}{2}}}\left[  2n+2\alpha+2\right]_{q^{\frac{1}{2}}}=\left[  2\right]_{q^{\frac{1}{2}}}
\left[  4n+2\alpha+4\right]_{q^{\frac{1}{2}}}.
$$
By the identity $\left[  2x\right]_{q^{\frac{1}{2}}}= \left[  2\right]_{q^{\frac{1}{2}}}\left[  x\right]_{q}$, we obtain
$$ \left[  4n+2\alpha+2\right]_{q^{\frac{1}{2}}}=\left[  2\right]_{q^{\frac{1}{2}}}\left[  2n+\alpha+1\right]_{q},$$
$$ \left[  4n+2\alpha+4\right]_{q^{\frac{1}{2}}}=\left[  2\right]_{q^{\frac{1}{2}}}
\left[  2n+\alpha+2\right]_{q},
$$
which leads to the commutation relations:
$$[K_0,\pm K]=\pm K_{\pm},\;\;\;[K_-, K_+]=\left[ 2K_0 \right]_{q}\;\;\;\mbox{on}\;\;\mathfrak{S}_{q\alpha}$$
and the conjugation relations
$$K_0^*=K_0,\;\;\; K_+^*=K_- \;\;\;\mbox{on}\;\;\mathfrak{S}_{q\alpha}.$$
To analyse irreducible representations  of the $ \mathsf{su}_{q^{\frac{1}{2}}}(1, 1)$  algebra,   we need
the invariant Casimir operator $C$,   which in this case has the  form:
$$ C= \left[ K_0-\frac{1}{2}\right]_{q}^2-K_+K_-.$$
From (\ref{B0B+B}) and (\ref{BB+act}) we obtain
the action  of this  operator on the  basis $\{  \phi_{n}^\alpha(x;q)   \}_{n=0}^\infty $:
$$C\phi_{2n}^\alpha(x;q)=\left(\left[ n +\frac{\alpha}{2}\right]_{q}^2-\left[ n\right]_{q}\left[ n+\alpha\right]_{q}\right) \phi_{2n}^\alpha(x;q),$$
$$C\phi_{2n+1}^\alpha(x;q)=\left(\left[ n+\frac{\alpha+1}{2}\right]_{q}^2-\left[ n\right]_{q}\left[ n+\alpha+1\right]_{q}\right) \phi_{2n+1}^\alpha(x;q).$$
Using (\ref{q-numb-add}) with $a= n+\frac{\alpha}{2}$, $b=n$, $c=-\frac{\alpha}{2}$  and
$a= n+\frac{\alpha+1}{2}$, $b=n$, $c=-\frac{\alpha+1}{2}$ respectively, we derive
$$\left[ n+\frac{\alpha}{2}\right]_{q}^2-\left[ n\right]_{q}\left[ n+\alpha\right]_{q}=\left[\frac{\alpha }{2}\right]_{q}^2 ,$$
$$\left[ n+\frac{\alpha+1}{2}\right]_{q}^2-\left[ n\right]_{q}\left[ n+\alpha+1\right]_{q}=\left[\frac{\alpha  +1}{2}\right]_{q}^2 .$$
% and commutes with $B_0$ and  $B_\pm$.\\
% In the q-oscillator Fock space
The Casimir operator $C$ has two  eigenvalues
$ \left[\frac{2\alpha+1  \mp 1}{4}\right]_{q}^2 $
in the subspaces $\mathfrak{S}_{q\alpha}^\pm$
% $ \mathcal{F}_{q,\alpha}^\pm$
formed by the  even and odd  basis vectors $\{  \phi_{n}^\alpha(x;q)   \}_{n=0}^\infty $, respectively.
Thus  $ \mathfrak{S}_{q\alpha}$ splits into the direct sum of two
$\mathsf{su}_{q^{\frac{1}{2}}}(1, 1)$-irreducible subspaces $ \mathfrak{S}_{q\alpha}^+$ and $\mathfrak{S}_{q\alpha}^-$.
\\
\begin{rem}
We deduce from (\ref{bb+actions}) that  the operators $b$ , $b^+$  and $N$ satisfy the relations
\begin{equation}\label{commut}
bb^+ - q^{\pm \frac{1+2\nu K}{2}}b^+b= \left[ 1+2\nu K \right]_{q^{\frac{1}{2}}}q^{\mp \frac{N+\nu -\nu K}{2}}
\;\;\;\mbox{on}\;\;\mathfrak{S}_{q\alpha},
\end{equation}
where $\nu=\alpha+\frac{1}{2}$.
This leads to  explicit expressions for the generators $\{b\;,b^+\;,N\}$ of the $q$-deformed Calogero$-$Vasiliev Oscillator algebra (see \cite{Macfarlane,Macfarlane1}) . In particular, Macfarlane in \cite{Macfarlane} has shown that
if $\nu= \frac{p-1}{2}$, this oscillator realises the $q$-deformed para-Bose oscillator of order $p$.
\end{rem}

{\bf ACKNOWLEDGMENTS.} The authors would like to express their sincere thanks
to the anonymous referee for a very careful reading of the manuscript and useful suggestions,
leading to Theorem 3.4, which has grately improved the quality of the paper.
This project was supported by King Faisal University  under Research Grant $\#$RG186116.

%--------------------------------------------%\bibliographystyle{amsplain}
%\bibliography{}

\end{document}